\documentclass[12pt]{amsart}

\usepackage{amssymb}
\usepackage{xcolor}
\usepackage{amsmath,epic,curves,amscd}
\usepackage{graphicx}
\usepackage{comment}
\usepackage{appendix}
\usepackage{mathdots}
\usepackage[all]{xy}
\usepackage[english]{babel}

\usepackage{enumitem}
\setlist[enumerate,1]{label=(\roman*)}

\newtheorem{claim}{}[section]
\newtheorem{theorem}[claim]{Theorem}

\newtheorem{conjecture}[claim]{Conjecture}
\newtheorem{lemma}[claim]{Lemma}
\newtheorem{proposition}[claim]{Proposition}

\renewenvironment{proof}{\noindent{\it Proof. \hskip0pt}}
                     {\hfill$\square$\par\medskip}

\begin{document}

\def\red{\textcolor{red}}

\newcommand\lan{\langle}
\newcommand\ran{\rangle}
\newcommand\tr{{\text{\rm Tr}}\,}
\newcommand\ot{\otimes}
\newcommand\join{\vee}
\newcommand\meet{\wedge}
\renewcommand\ker{{\text{\rm Ker}}\,}
\newcommand\image{{\text{\rm Im}}\,}
\newcommand\id{{\text{\rm id}}}
\newcommand\tp{{\text{\rm tp}}}
\newcommand\pr{\prime}
\newcommand\e{\epsilon}
\newcommand\la{\lambda}
\newcommand\inte{{\text{\rm int}}\,}
\newcommand\ttt{{\text{\rm t}}}
\newcommand\spa{{\text{\rm span}}\,}
\newcommand\conv{{\text{\rm conv}}\,}
\newcommand\rank{\ {\text{\rm rank of}}\ }
\newcommand\re{{\text{\rm Re}}\,}
\newcommand\ppt{\mathbb T}
\newcommand\rk{{\text{\rm rank}}\,}
\newcommand\SN{{\text{\rm SN}}\,}
\newcommand\SR{{\text{\rm SR}}\,}
\newcommand\HA{{\mathcal H}_A}
\newcommand\HB{{\mathcal H}_B}
\newcommand\HC{{\mathcal H}_C}
\newcommand\CI{{\mathcal I}}
\newcommand{\bra}[1]{\langle{#1}|}
\newcommand{\ket}[1]{|{#1}\rangle}
\newcommand\cl{\mathcal}
\newcommand\idd{{\text{\rm id}}}
\newcommand\OMAX{{\text{\rm OMAX}}}
\newcommand\OMIN{{\text{\rm OMIN}}}
\newcommand\diag{{\text{\rm Diag}}\,}
\newcommand\calI{{\mathcal I}}
\newcommand\bfi{{\bf i}}
\newcommand\bfj{{\bf j}}
\newcommand\bfk{{\bf k}}
\newcommand\bfl{{\bf l}}
\newcommand\bfp{{\bf p}}
\newcommand\bfq{{\bf q}}
\newcommand\bfzero{{\bf 0}}
\newcommand\bfone{{\bf 1}}
\newcommand\im{{\mathcal R}}
\newcommand\ha{{\frac 12}}
\newcommand\xx{{\text{\sf X}}}
\newcommand\sa{{\text{\rm sa}}}
\newcommand\Pz{P^{\rm z}}
\newcommand\Pn{P^{\rm n}}
\newcommand\textfrac{\textstyle\frac}
\newcommand\semi{}

\title{Entangled edge states of corank one with positive partial transposes}

%
%
%

\author{Jinwon Choi}
\address{Department of Mathematics and Research Institute of Natural Sciences, Sookmyung Women's University, Seoul 04310, Korea}
\email{jwchoi@sookmyung.ac.kr}

\author{Young-Hoon Kiem}
\address{Department of Mathematics and Research Institute of Mathematics, Seoul National University, Seoul 08826, Korea}
\email{kiem@snu.ac.kr}

\author{Seung-Hyeok Kye}
\address{Department of Mathematics and Research Institute of Mathematics, Seoul National University, Seoul 08826, Korea}
\email{kye@snu.ac.kr}

\date{\today}

\begin{abstract}
We construct a parameterized family of $n\otimes n$ PPT (positive partial transpose) states of corank one for each $n\ge 3$.
With a suitable choice of parameters, we show that they are $n\otimes n$ PPT
entangled edge states of corank one for $3\le n\le 1000$.
They violate the range criterion for separability in the most extreme
way. Note that corank one is the smallest possible corank for such
states. The corank of the partial transpose is given by $2n-3$,
which is also the smallest possible corank for the partial
transposes of PPT entangled edge states of corank one. They provide
the first explicit examples of such states for $n\ge 4$.
\end{abstract}

\maketitle
\def\al{\alpha }
\def\bal{\bar{\alpha} }
\def\rhog{{\varrho^\Gamma} }
\def\im{\mathrm{Im}\, }
\def\bx{\bar{x} }
\def\CC{\mathbb{C} }
\def\beq{\begin{equation}}
\def\eeq{\end{equation}}

\section{Introduction}

In the current quantum information and computation theory, the notion of entanglement
is considered as one of the most important resources. Nevertheless, distinguishing
entanglement from separability is very difficult, and known to be NP-hard in general \cite{{StrongNP},{NPHard03}}.
Among various separability criteria, the PPT (positive partial transpose) criterion \cite{choi-ppt, peres} is very simple to test but powerful:
The partial transpose of a separable state must be positive (semi-definite).
Positivity of the partial transpose is actually sufficient for separability
in the $2\ot 2$ and $2\ot 3$ systems \cite{horo-1,stormer,woronowicz}, but this is not the case
in general. Examples of PPT entanglement go back to the seventies and early eighties:
See \cite{woronowicz}  for $2\otimes 4$ case and  \cite{choi-ppt,{stormer82}} for $3\otimes 3$ case.
The notion of PPT is also very important in itself in quantum information theory.
See \cite{{halder},{horo-distill},{huber}} for examples.

A state in the tensor product $M_m\ot M_n$ of matrix
algebras is called \emph{separable} if it is a convex combination of pure
product states, which are rank one projections onto product vectors
of the form $|\xi\ran\ot |\eta\ran$ in $\mathbb C^m\ot\mathbb C^n$.
Non-separable states are called \emph{entangled}. The \emph{partial transpose}
$(x\ot y)^\Gamma$ of $x\ot y\in M_m\ot M_n$ is given by $x^\ttt\ot y$ with the usual
transpose $x^\ttt$. If we identify $M_m\ot M_n$ with the block
matrices $M_m(M_n)$ then the partial transpose corresponds to the
block-wise transpose. Recall that the transpose of the rank one
projection $|\xi\ran\lan\xi|$ onto $|\xi\ran$ is again a rank one
projection onto its conjugate vector $|\bar\xi\ran$. Therefore, if a
PPT state $\varrho$ is separable then there must exist a family
$\{|\xi_i\ran\ot |\eta_i\ran\}$ of product vectors such that the
ranges of $\varrho$ and $\varrho^\Gamma$ are spanned by
$\{|\xi_i\ran\ot |\eta_i\ran\}$ and $\{|\bar\xi_i\ran\ot
|\eta_i\ran\}$, respectively. This is the range criterion for
separability \cite{p-horo-range} which is useful to detect
entanglement among PPT states.

Some PPT entangled states violate the range criterion in an extreme way: There exists no
nonzero product vector $|\xi_i\ran\ot |\eta_i\ran\in \image\varrho$ such that
$|\bar\xi_i\ran\ot |\eta_i\ran\in \image\varrho^\Gamma$. Such states are called PPT entangled edge states \cite{lkhc},
in short, \emph{edge states} in this paper. Edge states are very important to understand the convex set consisting of PPT states,
because every PPT state is a convex combination of pure product states and edge states.
It is clear that there must be some restrictions on the ranges of edge states and their partial transposes, and so
it is natural to classify edge states by the {\sl bi-ranks}, $(p,q)$, combinations of ranks
$p$ and $q$ of themselves and their partial transposes, respectively.
The first step for classification is to solve related equations to get necessary conditions
for possible bi-ranks. The next step is, of course, to construct PPT entangled edge states with prescribed possible bi-ranks.
The first step has been considered in \cite{kye-prod-vec} with techniques from algebraic geometry. See also
\cite{kiem-multi} for multi-partite cases.

In the $3\ot 3$ system, the classification of edge states by bi-ranks is now complete
by constructing \cite{kye-osaka} $3\ot 3$ edge states of bi-rank $(8,6)$, together with others
\cite{{bdmsst},{choi-ppt},{clarisse},{dmsst},{ha-3},{ha-kye-2},{stormer82}}.
Note that the sum $8+6=14$ is the maximum among $p+q$ of possible bi-rank $(p,q)$ for $3\ot 3$ edge states.
It is still an open question to classify $2\ot 4$ edge states by bi-ranks. See \cite{kye-ritsu} for a survey in this direction.
In the general $n\ot n$ system, possible maximum bi-rank for edge states is known to be $(n^2-1, n^2-2n+3)$ \cite{kye-prod-vec}.
The purpose of this paper is to construct such PPT entangled edge states. We construct PPT states
with bi-rank $(n^2-1, n^2-2n+3)$ for arbitrary $n=3,4,\dots$, and confirm that some of them are
PPT entangled edge states up to $n\le 1000$. In fact, we conjecture that our PPT states are
entangled edge states for all $n\ge 3$ with a suitable choice of parameters.

We explain the background why $(n^2-1,n^2-2n+3)$ is a maximum
bi-rank for possible edge states and provide positive \semi matrices
for our construction in the next section, and then we solve a system
of bilinear equations that will determine the ranges of our states
in Section \ref{bilinear-equation}. We warm up by constructing $3\ot
3$ and $4\ot 4$ edge states in Sections \ref{3x3} and \ref{4x4},
respectively, and present general construction in Section \ref{nxn}. In the final section, we
discuss alternative constructions for $4\otimes 4$ PPT entangled
edge states of corank one.

\section{PPT entangled edge states with minimum coranks}

We are looking for quadruplet $(m,n,k,\ell)$ of natural numbers satisfying the following property:
\smallskip
\begin{enumerate}
\item[{\rm (A)}]
there exists an $m\otimes n$ PPT entangled edge state $\varrho$ with bi-rank $(mn-k,mn-\ell)$.
\end{enumerate}
The numbers $k$ and $\ell$ are called the \emph{corank} of the matrix $\varrho$ and $\varrho^\Gamma$, respectively. If the statement (A) is true, then we may take subspaces $D=\ker\varrho$ and $E=\ker\varrho^\Gamma$ to see that the following property
\begin{enumerate}
\item[{\rm (B)}]
there exists a pair $(D,E)$ of subspaces of $\mathbb C^m\ot\mathbb C^n$ with $(\dim D,\dim E)=(k,\ell)$ such that
there is no nonzero product vector $|\xi\ran\ot |\eta\ran$ satisfying
\begin{equation}\label{bas-eq}
|\xi\ran\ot |\eta\ran\in D^\perp \quad \text{and}\quad |\bar\xi\ran\ot |\eta\ran\in E^\perp
\end{equation}
\end{enumerate}
holds. Note that (\ref{bas-eq}) gives rise to a system of equations with
$m+n-2$ complex variables up to scalar multiplications. Because the total number of equations is given by
$k+\ell$, one may expect that the statement (B) implies that $k+\ell\ge m+n-2$. It was actually shown in \cite{kye-prod-vec}
that the statement (B) implies the following:
\begin{enumerate}
\item[{\rm (C)}]
$k+\ell> m+n-2$ or the following relation
\begin{equation}\label{kra}
k+\ell=m+n-2,\qquad \sum_{r+s=m-1}(-1)^r \binom kr\binom \ell s = 0
\end{equation}
holds.
\end{enumerate}

The Diophantine equation (\ref{kra}) is known as the Krawtchouk polynomial, which is originated
from harmonic analysis and plays an important role in
the current coding theory. See \cite{HeoK}, \cite{MWS} and \cite{vint}. Even though the equation (\ref{kra})
is not yet solved completely, there are several easy solutions. For example, in case of $m=n$ it is easy to see that
$(k,\ell)$ satisfies (\ref{kra}) if and only if both $k$ and $\ell$ are odd. In other word,
$(n,n,k,2n-k-2)$ is a solution of (\ref{kra}) for every odd number $k$ with $1\le k< n$.
Therefore, it is natural to ask whether (A) holds for these quadruplets or not. The case of $k=1$ with
the quadruplet $(n,n,1,2n-3)$ satisfying ({\ref{kra}) is of special interest,
because this gives rise to the minimum corank one for edge states, together with the minimum corank
of the partial transposes of edge states of corank one.

We think of an $n^2\times n^2$ matrix $\varrho$ in $M_n\otimes M_n$ as an $n\times n$ block
matrix, each of whose blocks $\varrho_{ij}$ is an $n\times n$ matrix.
Its partial transpose $\varrho^\Gamma$ is the result of swapping the
$(i,j)$th block with the $(j,i)$th block, that is, the $(i,j)$th block
of $\varrho^\Gamma$ is $\varrho_{ji}$. Let $\{e_1, e_2, \cdots, e_n\}$
denote the standard basis for $\CC^n$. We will use
$e_{ij}=e_i\otimes e_j$ with lexicographic ordering as the basis for
$\CC^n\otimes \CC^n$, so the rows and the columns of the $n^2 \times n^2$ matrix $\varrho$ are indexed by $e_{ij}$ in lexicographic order. 

It is clear that PPT states with corank zero are never edge states. It is also easy to construct
$n\ot n$ PPT states of corank one. For example, we take an $n\times n$ positive \semi matrix $A$ with corank one,
and consider an $n^2\times n^2$ matrix whose $(e_{11}, e_{22},\dots, e_{nn})$-principal submatrix is given by
$A$, together with suitably chosen diagonal entries, while all the other entries are zeros. Here, $(e_{11}, e_{22},\dots, e_{nn})$-principal submatrix means the submatrix of $A$ consisting of rows and columns indexed by $e_{11}, e_{22},\dots, e_{nn}$. 
After taking the partial transpose, it is not hard to see that the matrix $\varrho^\Gamma$ can be decomposed into principal submatrices of sizes $2\times 2$ or $1\times 1$. Indeed, the $(e_{ij},e_{ji})$-principal submatrix of $\varrho^\Gamma$ has size $2 \times 2$ and all the other nonzero entries are diagonal entries. Thus, one can choose the diagonal entries of $\varrho$ so that $\varrho^\Gamma$ is positive. 
This is in fact basically how Choi \cite{choi-ppt} and St\o rmer \cite{stormer82} constructed
special kinds of block matrices in $M_3(M_3)$,
which turn out to be $3\ot 3$ PPT entangled edge states with bi-ranks $(4,4)$ and $(6,7)$, respectively.
The same idea has been adopted to construct $3\ot 3$ edge states with bi-rank $(8,6)$ in \cite{kye-osaka}.
We note that the number of $2\times 2$ principal submatrices coincides with the corank of the partial transpose
in this construction of $3\otimes 3$ PPT states.

If we follow the above idea for $n\ot n$ cases with $n\ge 4$ then the number $n(n-1)/2$ of $2\times 2$ principal matrices
of $\varrho^\Gamma$ exceeds $2n-3$. We overcome this situation by using the following $d\times d$ matrices
\beq\label{1.10}P_d(z_2,\cdots,z_d)=\left( \begin{matrix}
2 & z_2 & 0 & 0 & \cdots &z_d\\
z_2^{-1} & 2 & z_3z_2^{-1} & 0 & \cdots &0\\
0&z_3^{-1}z_2 &2 &z_4z_3^{-1}&\cdots & 0\\
0&0&z_5^{-1}z_2  &2 &\ddots &\vdots\\
\vdots & \vdots &\ddots &\ddots &\ddots &\vdots\\
z_d^{-1}&0&0&\cdots & z_d^{-1}z_{d-1} &2
\end{matrix}\right)\in M_d\eeq
instead of $2\times 2$ matrices, as building blocks for the partial transpose $\varrho^\Gamma$ of edge states $\varrho$, where
$d\ge 4$ is an even integer and $z_j\in \CC$ with $|z_j|=1$ for $j=2,\cdots,d$.
The matrix $P_d(z_2,\cdots,z_d)$ is positive \semi of corank one. In fact, when $z_i=(-1)^{i-1}$ for  $2\le i\le d$,
we see that $P_d(-1,1,\cdots,-1)$ is the Cartan matrix $(2\delta_{ij}-a_{ij})$ of the graph with $d$ vertices and $d$ edges that form a cycle.
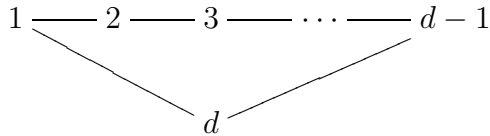
\[\xymatrix{
1\ar@{-}[r] & 2\ar@{-}[r] & 3\ar@{-}[r] & \cdots \ar@{-}[r] & d-1\\
&&d\ar@{-}[llu]\ar@{-}[urr]
}\]
Here $(a_{ij})$ is the adjacency matrix defined by $a_{ij}=1$ if the vertices $i$ and $j$ are connected by an edge and $a_{ij}=0$ if not.
It is well known that this Cartan matrix (of affine type) is positive \semi of corank one.
Since
\[P_d(z_2,\cdots,z_d)= \diag(1,-z_2,z_3,\cdots,-z_d)^{*} P_d(-1,1,\cdots,-1)\diag(1,-z_2,z_3,\cdots,-z_d), \]
we see that $P_d(z_2,\cdots,z_d)$ is always positive \semi of corank one.
The kernel of $P_d(z_2,\cdots,z_d)$ is spanned by the vector
\beq\label{1.11} (1,-z_2^{-1},z_3^{-1},-z_4^{-1},\cdots,-z_d^{-1})^\ttt\in\mathbb C^d.\eeq
For $d=2$ and $z\in \CC$ with $|z|=1$, we will use
$$
P_2(z)=\left( \begin{matrix}
1& z\\
z^{-1} & 1
\end{matrix}\right)\in M_2
$$
which is positive \semi of corank one with kernel spanned by $(1,-z^{-1})^\ttt.$

We will take the partial transposes of $\varrho^\Gamma$ to get edge states $\varrho=(\varrho^\Gamma)^\Gamma$ of corank one,
which have principal matrices of the form
\beq\label{1.15}Q_n(z_2,\cdots,z_n)=\left( \begin{matrix}
2 & z_2 & 0 & 0 & \cdots &0\\
z_2^{-1} & 2 & z_3 & 0 & \cdots &0\\
0&z_3^{-1} &2 &z_4&\cdots & 0\\
\cdots & \cdots &\cdots &\cdots &\cdots &\cdots\\
0&\cdots&0&z_{n-1}^{-1}&2&z_n\\
0&0&\cdots& 0& z_n^{-1} &2
\end{matrix}\right)\in M_n.\eeq
By row expansion and induction,
we see that this is positive definite for $z_i\in \CC$ with $|z_i|=1$.
In fact, the determinant of $Q_n(z_2,\cdots,z_n)$ is precisely $n+1$.

\section{bilinear equations}\label{bilinear-equation}

For a given $\alpha=(\alpha_1,\alpha_2,\dots,\alpha_n)\in\mathbb
C^n$ with nonzero entries $\alpha_i\neq 0$, we consider the bilinear
equation $[i,j]_\alpha$ with unknowns $x=(x_1,\cdots,x_n)$, $y=(y_1,\cdots,y_n)\in\mathbb C^n$, defined
by
$$
[i,j]_\alpha=x_iy_j-\alpha_i^{-1}\alpha_jx_jy_i,
$$
for $i,j=1,2,\dots,n$.

\begin{lemma}\label{basic}
Suppose that $i$, $j$ and $k$ are mutually distinct. If
$[i,j]_\alpha=[i,k]_\alpha=0$ with $(x_i,y_i)\neq0$ then
$[j,k]_\alpha=0$.
\end{lemma}

\begin{proof}
Since the system
$$\left( \begin{matrix}
y_j & -\alpha_{i}^{-1}\al_j x_j \\ y_k & -\alpha_{i}^{-1}\al_k x_k
\end{matrix}\right) \left(\begin{matrix} x_i\\
y_i\end{matrix}\right)=\left(\begin{matrix} 0\\
0\end{matrix}\right)$$ of linear equations $[i,j]_{\alpha}=0$ and
$[i,k]_{\alpha}=0$ in $(x_i,y_i)$ has a nontrivial solution, the
determinant of the above $2\times 2$ matrix is zero, and hence we
have
$$
[j,k]_{\alpha}=x_jy_k-\alpha_{j}^{-1}\al_k
x_ky_j=\al_i\al_j^{-1}(\al_i^{-1}\al_j
x_jy_k-\al_i^{-1}\al_kx_ky_j)=0,
$$
as it is required.
\end{proof}

In this section, we fix $n=3,4,\dots$, and solve the system
\begin{equation}\label{eq}
\begin{array}{lll}
&[1,k]_\alpha- [2,k-1]_\alpha+[3,k-2]_\alpha-\cdots
+(-1)^{\lfloor\textfrac{k}{2}\rfloor-1}[\lfloor\textfrac{k}{2}\rfloor, k+1-\lfloor \textfrac{k}{2}\rfloor
]_\alpha =0,
\\
&[n-\ell,n]_{\beta}-[n-\ell+1,n-1]_{\beta}+\cdots
+(-1)^{\lfloor\textfrac{\ell-1}{2}\rfloor}[n-\ell+\lfloor\textfrac{\ell-1}{2}\rfloor,
n-\lfloor\textfrac{\ell-1}{2}\rfloor]_{\beta}=0,
\end{array}
\end{equation}
of equations with $k=2,3,\dots, n$ and $\ell=1,2,\dots,n-2$, where
$\alpha,\beta\in\mathbb C^n$ have no zero entry. Here, $\lfloor s\rfloor$ denotes the greatest integer less than or equal to $s$.
When $n=3$ and $n=4$, (\ref{eq}) becomes
$$
[1,2]_\alpha=0,\quad [1,3]_\al=0,\quad [2,3]_\beta= 0
$$
and
$$
[1,2]_\alpha=0,\quad [1,3]_\al=0,\quad [1,4]_\al-[2,3]_\al=0,\
[3,4]_\beta=0,\ [2,4]_\beta=0,
$$
respectively. When $n=5$ and $n=6$, (\ref{eq}) tells us that the
following forms
$$
[1,2]_\alpha,\ [1,3]_\al,\ [1,4]_\al-[2,3]_\al,\
[1,5]_\al-[2,4]_\al,\ [4,5]_\beta,\ [3,5]_\beta,\
[2,5]_\beta-[3,4]_\beta
$$
and
$$
\begin{aligned}
&[1,2]_\alpha,\ [1,3]_\al,\ [1,4]_\al-[2,3]_\al,\ [1,5]_\al-[2,4]_\al,\ [1,6]_\al -[2,5]_\al +[3,4]_\al,\\
&[5,6]_\beta,\ [4,6]_\beta,\ [3,6]_\beta-[4,5]_\beta,\
[2,6]_\beta-[3,5]_\beta
\end{aligned}
$$
are zeros, respectively. Figure 1 shows which $[j,k]_\alpha$ and
$[j,k]_\beta$ appear in the equation (\ref{eq}). The following lemma
shows that all such $[j,k]_\alpha$ and $[j,k]_\beta$ must be zero.

\newcommand\cii{\circle*{0.4}}
\begin{figure}
\begin{center}
\setlength{\unitlength}{.5 truecm}
\begin{picture}(27,8)
\put(1,4){\line(1,0){2}} \put(1,5){\line(1,0){2}}
\put(1,6){\line(1,0){2}} \put(1,4){\line(0,1){2}}
\put(2,4){\line(0,1){2}} \put(3,4){\line(0,1){2}} \put(0.5,5.7){$1$}
\put(0.5,4.7){$2$} \put(0.5,3.7){$3$} \put(0.5,2.7){$j$}
\put(0.8,6.3){$1$} \put(1.8,6.3){$2$} \put(2.8,6.3){$3$}
\put(3.8,6.3){$k$} \put(2,6){\cii}\put(3,6){\cii}
\put(3,5){\circle{0.4}}

\put(5,3){\line(1,0){3}} \put(5,4){\line(1,0){3}}
\put(5,5){\line(1,0){3}} \put(5,6){\line(1,0){3}}
\put(5,3){\line(0,1){3}} \put(6,3){\line(0,1){3}}
\put(7,3){\line(0,1){3}} \put(8,3){\line(0,1){3}}
\put(6,6){\cii}\put(7,6){\cii}\put(8,6){\cii}\put(7,5){\cii}
\put(8,5){\circle{0.4}}\put(8,4){\circle{0.4}}

\put(9,2){\line(1,0){4}} \put(9,3){\line(1,0){4}}
\put(9,4){\line(1,0){4}} \put(9,5){\line(1,0){4}}
\put(9,6){\line(1,0){4}} \put(9,2){\line(0,1){4}}
\put(10,2){\line(0,1){4}} \put(11,2){\line(0,1){4}}
\put(12,2){\line(0,1){4}} \put(13,2){\line(0,1){4}}
\put(10,6){\cii}\put(11,6){\cii}\put(12,6){\cii}\put(13,6){\cii}\put(11,5){\cii}\put(12,5){\cii}
\put(13,5){\circle{0.4}}\put(13,4){\circle{0.4}}\put(13,3){\circle{0.4}}\put(12,4){\circle{0.4}}

\put(14,1){\line(1,0){5}} \put(14,2){\line(1,0){5}}
\put(14,3){\line(1,0){5}} \put(14,4){\line(1,0){5}}
\put(14,5){\line(1,0){5}} \put(14,6){\line(1,0){5}}
\put(14,1){\line(0,1){5}} \put(15,1){\line(0,1){5}}
\put(16,1){\line(0,1){5}} \put(17,1){\line(0,1){5}}
\put(18,1){\line(0,1){5}} \put(19,1){\line(0,1){5}}
\put(15,6){\cii}\put(16,6){\cii}\put(17,6){\cii}\put(18,6){\cii}\put(19,6){\cii}
\put(16,5){\cii}\put(17,5){\cii}\put(18,5){\cii}\put(17,4){\cii}
\put(19,5){\circle{0.4}}\put(19,4){\circle{0.4}}\put(19,3){\circle{0.4}}\put(19,2){\circle{0.4}}
\put(18,4){\circle{0.4}}\put(18,3){\circle{0.4}}

\put(20,0){\line(1,0){6}} \put(20,1){\line(1,0){6}}
\put(20,2){\line(1,0){6}} \put(20,3){\line(1,0){6}}
\put(20,4){\line(1,0){6}} \put(20,5){\line(1,0){6}}
\put(20,6){\line(1,0){6}} \put(20,0){\line(0,1){6}}
\put(21,0){\line(0,1){6}} \put(22,0){\line(0,1){6}}
\put(23,0){\line(0,1){6}} \put(24,0){\line(0,1){6}}
\put(25,0){\line(0,1){6}} \put(26,0){\line(0,1){6}}
\put(21,6){\cii}\put(22,6){\cii}\put(23,6){\cii}\put(24,6){\cii}\put(25,6){\cii}\put(26,6){\cii}
\put(22,5){\cii}\put(23,5){\cii}\put(24,5){\cii}\put(25,5){\cii}\put(23,4){\cii}\put(24,4){\cii}
\put(26,5){\circle{0.4}}\put(26,4){\circle{0.4}}\put(26,3){\circle{0.4}}\put(26,2){\circle{0.4}}\put(26,1){\circle{0.4}}
\put(25,4){\circle{0.4}}\put(25,3){\circle{0.4}}\put(25,2){\circle{0.4}}\put(24,3){\circle{0.4}}
\end{picture}
\end{center}
\caption{Bullets and circles represent the positions $(j,k)$ for
which the forms $[j,k]_\alpha$ and $[j,k]_\beta$ appear in the
system (\ref{eq}) of equations, respectively, for $n=3,4,5,6,7$. }
\end{figure}
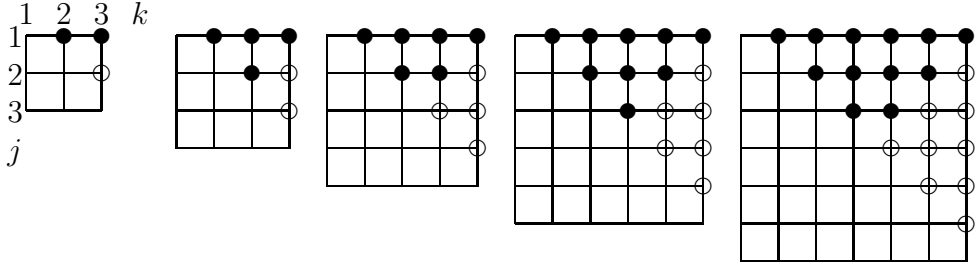
\medskip

\begin{lemma}\label{1.3}
If $x,y\in\mathbb C^n$ satisfy {\rm
(\ref{eq})}, then we have
\begin{equation}\label{1.4}
\begin{array}{lll}
&[j,k]_{\alpha}=0,\quad &\text{for all}\ (j,k)\ \text{with}\ k\ge j+1 \text{ and } k\le n+1-j,\\
&[j,k]_{\beta}=0,\quad &\text{for all}\  (j,k)\ \text{with}\ k\ge
j+1 \text{ and } k> n+1-j.
\end{array}
\end{equation}
\end{lemma}

\begin{proof}
First we prove that $[1,k]_{\al}=0$ for $k=2,3,\dots,n$. When
$(x_1,y_1)=(0,0)$, this is trivial. So we may assume $(x_1,y_1)\ne
(0,0)$. Then $[1,2]_\alpha=[1,3]_\alpha=0$ implies $[2,3]_\alpha=0$
by Lemma \ref{basic}. Then by $[1,4]_\al-[2,3]_\al=0$, we have
$[1,4]_{\al}=0.$ From the system $[1,2]_{\alpha}=0=[1,4]_{\al}$, we
also have $[2,4]_{\al}=0$, and hence $[1,5]_{\al}=0$ by
$[1,5]_\al-[2,4]_\al=0$. Continuing in this way, we find that
$[1,k]_{\al}=0$ for each $k=2,3,\dots,n$. Deleting $[1,k]_{\al}$
from \eqref{eq} and \eqref{1.4}, we are in the situation with less variables.
Induction on $n$ completes the proof for $[j,k]_\alpha$. The exactly
same argument can be applied for $[j,k]_\beta$.
\end{proof}

In this paper, we assume the following:
\beq \al_i^{-1}\al_j\ne \beta_i^{-1}\beta_j\quad \text{for }1\le i<j\le n.\eeq
For a given solution $x,y\in\mathbb C^n$ of the system \eqref{eq}, we put
$v_i=(x_i,y_i)\in\mathbb C^2$ for $i=1,2,\dots,n$. From now on, we
assume that both $x$ and $y$ are nonzero, and denote by $p$ and $q$
the smallest and largest number $i=1,2,\dots,n$ so that $v_i\neq 0$,
respectively. Then we have the following four cases (See Figure 2.):
\begin{itemize}
\item
$1\le p\le q\le\textfrac{n+1}2$;
\item
$1\le p < \textfrac{n+1}2<q\le n$ and $p+q\le n+1$;
\item
$1\le p < \textfrac{n+1}2<q\le n$ and $p+q> n+1$;
\item
$\textfrac{n+1}2  \le p\le q\le n$.
\end{itemize}

\begin{figure}
\begin{center}
\setlength{\unitlength}{.5 truecm}
\begin{picture}(28,7)
\put(0,0){\line(1,0){6}}\put(0,0){\line(0,1){6}}
\put(6,6){\line(-1,0){6}}\put(6,6){\line(0,-1){6}}
\put(7,0){\line(1,0){6}}\put(7,0){\line(0,1){6}}
\put(13,6){\line(-1,0){6}}\put(13,6){\line(0,-1){6}}
\put(14,0){\line(1,0){6}}\put(14,0){\line(0,1){6}}
\put(20,6){\line(-1,0){6}}\put(20,6){\line(0,-1){6}}
\put(21,0){\line(1,0){6}}\put(21,0){\line(0,1){6}}
\put(27,6){\line(-1,0){6}}\put(27,6){\line(0,-1){6}}
\put(6,0){\line(-1,1){6}}\put(13,0){\line(-1,1){6}}\put(20,0){\line(-1,1){6}}\put(27,0){\line(-1,1){6}}
\put(6,6){\line(-1,-1){3}}\put(13,6){\line(-1,-1){3}}\put(20,6){\line(-1,-1){3}}\put(27,6){\line(-1,-1){3}}
\linethickness{2pt}
\put(0.5,3.5){\line(1,0){2}}\put(0.5,3.5){\line(0,1){2}}\put(2.5,5.5){\line(-1,0){2}}\put(2.5,5.5){\line(0,-1){2}}
\put(8,2){\line(1,0){3}}\put(8,2){\line(0,1){3}}\put(11,5){\line(-1,0){3}}\put(11,5){\line(0,-1){3}}
\put(16,1){\line(1,0){3}}\put(16,1){\line(0,1){3}}\put(19,4){\line(-1,0){3}}\put(19,4){\line(0,-1){3}}
\put(24.5,0.5){\line(1,0){2}}\put(24.5,0.5){\line(0,1){2}}\put(26.5,2.5){\line(-1,0){2}}\put(26.5,2.5){\line(0,-1){2}}
\put(-0.4,3.3){$q$}\put(-0.45,5.3){$p$}\put(0.4,6.2){$p$}\put(2.4,6.1){$q$}
\end{picture}
\end{center}
\caption{Four possible locations of the square $[p,q]\times [p,q]$.
To solve the equation (\ref{eq}), it is enough to consider the
equations \lq inside\rq\ the square. }
\end{figure}
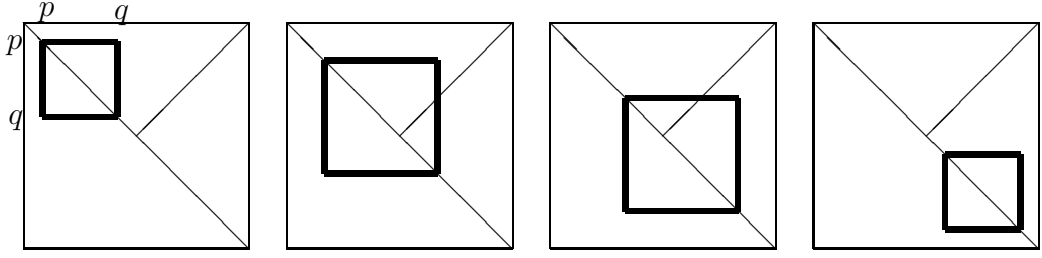
\medskip

In the first case, the equation (\ref{eq}) is reduced to
$\al_ix_iy_j=\al_jx_jy_i$ for every $i,j=p,p+1,\dots,q$, and so we
see that $(\al_px_p,\dots,\al_qx_q)$ is parallel to $(y_p,\dots,
y_q)$, and the solutions are given by $v_j=(c_jt,c_j\al_j)$ for
$j=p,\dots,q$ with $t\neq 0$.

In the second case, we first note $p\le n-q+1<q$. We will show
\begin{equation}\label{zero}
v_j=0\quad \text{for}\ n-q+1<j<q,
\end{equation}
which implies that $(\al_px_p,\al_{p+1}x_{p+1},\dots,
\al_{n-q+1}x_{n-q+1},\alpha_q x_q)$ and $(y_p,y_{p+1},\dots, y_{n-q+1},y_q)$
are parallel. We may suppose that $q-p>1$, because there is nothing
to prove when $q-p=1$. For each $p<j<q$, we have
$[p,j]_\al=[p,q]_\al=0$ by Lemma \ref{1.3}, and so we have
$[j,q]_\al=0$ by Lemma \ref{basic}. We also have $[j,q]_\beta=0$ by
Lemma \ref{1.3} again for $ n-q+1<j<q$. Since
$\alpha^{-1}_j\alpha_q\neq \beta^{-1}_j\beta_q$, we have
$x_jy_q=x_qy_j=0$ for $n-q+1<j<q$. Therefore,
\eqref{zero} follows once we prove
that $x_q\neq 0$ and $y_q\neq 0$.
If one of them is zero, say, $x_q=0$, then $y_q\neq 0$ implies that $x=0$
because $[j,q]_\al=0$ for all $p<j<q$. This contradicts the assumption that $x\neq 0$.

The third and fourth cases can be solved by the same ways as the
second and first cases, respectively. We summarize as follows:

\begin{lemma}\label{1.6}
Let $\alpha_i,\beta_i\in \CC^*$ and let $\al_{i,j}=\al_i^{-1}\al_j$
and $\beta_{i,j}=\beta_i^{-1}\beta_j$ for $1\le i, j\le n$. We assume $\al_{i,j}\ne \beta_{i,j}$ for $1\le
i<j\le n$. If $x,y\in \CC^n$ satisfy {\rm (\ref{eq})}, then one of the
following holds:
\begin{enumerate}
\item[{\rm (i)}]
$x=0$ or $y=0$;
\item[{\rm (ii)}]
$v_j=(c_jt,c_j\al_{j})$ for $t\in \CC^*$, $c_j\in \CC$ if $j\le
\frac{n+1}2$ and $v_j=0$ if $j>\frac{n+1}2$;
\item[{\rm (iii)}]
$v_1=\cdots=v_{p-1}=0$, $v_{n-q+2}=\cdots =v_{q-1}=0$,
$v_{q+1}=\cdots=v_n=0$ and $v_j=(c_jt,c_j\al_{j})$ for all $j$
with $p\le j\le n-q+1 < \frac{n+1}2$ or $j=q$, where $t\in \CC^*$,
$c_j\in \CC$, $c_p\ne 0$ and $c_{q}\ne 0$;
\item[{\rm (iv)}]
$v_1=\cdots=v_{p-1}=0$, $v_{p+1}=\cdots =v_{n-p+1}=0$,
$v_{q+1}=\cdots=v_n=0$ and
$v_{j}=(c_{j}t,c_{j}\beta_{j})$ for all $j$ with
$\frac{n+3}2 <n-p+2\le j\le q $ or $j=p$, where $t\in \CC^*$,
$c_{j}\in \CC$, $c_{p}\ne 0$ and $c_{q}\ne 0$;
\item[{\rm (v)}]
$v_j=(c_jt,c_j\beta_{j})$ for $t\in \CC^*$, $c_j\in \CC$ if $j\ge
\frac{n+1}2$ and $v_j=0$ if $j<\frac{n+1}2$. 

\end{enumerate}
\end{lemma}

Our strategy to construct a PPT entangled edge state $\varrho$ is as follows. We construct $\varrho^\Gamma$ which has $(2n-3)$ principal submatrices of the form $P_d(z_2,\cdots,z_d)$, where $z_i$'s are either $1$, $\alpha_{i,j}$ or $\beta_{i,j}$. These parameters are chosen so that $\bx\otimes y \in \im \rhog$ if and only if $x,y\in \CC^n$ satisfy {\rm (\ref{eq})} and hence Lemma \ref{1.6}. Next, we choose the diagonal entries of $\varrho^\Gamma$ which make $\varrho$ positive of corank one. Finally, we choose parameters $\alpha_i$'s and $\beta_i$'s so that the conditions $x\otimes y \in \im \rho$ and $\bx\otimes y \in \im \rhog$ imply that either $x=0$ or $y=0$. In the next two sections, we provide explicit examples in $3\otimes 3$ and $4\otimes 4$ system.

\section{$3\otimes 3$ edge states of corank one}\label{3x3}

Let $\al_i, \beta_i\in \CC$ with $|\alpha_i|=|\beta_i|=1$ for $1\le i\le 3$.
As above, we write $\al_{i,j}=\al_i^{-1}\al_j$ and $\beta_{i,j}=\beta_i^{-1}\beta_j$.
We also assume $\al_{i,j}\ne \beta_{i,j}$ for $1\le i<j\le 3$.

Let $\rhog$ be the $9\times 9$ matrix defined as follows:
\begin{enumerate}
\item The $(e_{12}, e_{21})$-principal submatrix is
$P_2(\al_{1,2})$.
\item The $(e_{13}, e_{31})$-principal submatrix is
$2P_2(\al_{1,3})$.
\item The $(e_{23}, e_{32})$-principal submatrix is
$P_2(\beta_{2,3})$.
\item The $(e_{11},e_{22},e_{33})$-principal submatrix is $rI_3$ for $r>1$ to be determined later.
\item All the other entries are zero.
\end{enumerate}
Namely we have
$$
\rhog=\left(\begin{array}{ccc|ccc|ccc}
r& \cdot&\cdot&  \cdot& \cdot&\cdot &   \cdot& \cdot&\cdot\\
\cdot& 1&\cdot&  \al_{1,2}& \cdot&\cdot   &\cdot& \cdot&\cdot\\
\cdot& \cdot&2 &\cdot& \cdot&\cdot &2\al_{1,3}& \cdot&\cdot\\
\hline
\cdot&\al_{1,2}^{-1}&\cdot&  1& \cdot&\cdot &\cdot& \cdot&\cdot\\
\cdot& \cdot&\cdot &   \cdot& r&\cdot &\cdot& \cdot&\cdot\\
\cdot& \cdot&\cdot &   \cdot& \cdot&1 &\cdot& \beta_{2,3}&\cdot\\
\hline
\cdot& \cdot&2\al_{1,3}^{-1} &\cdot& \cdot&\cdot &2& \cdot&\cdot\\
\cdot& \cdot&\cdot &\cdot& \cdot&\beta_{2,3}^{-1} &\cdot& 1&\cdot\\
\cdot& \cdot&\cdot &\cdot& \cdot&\cdot &\cdot& \cdot&r
\end{array}\right),
$$
where $\cdot$ denotes zero.

Since the matrix $P_2(z)$ of $|z|=1$ has corank one, it is easy to see that $\rhog$ is positive \semi of corank 3.
The kernel is spanned by the vectors
$$
e_{12}-\bar\al_{1,2}e_{21},\qquad
e_{13}-\bar\al_{1,3}e_{31},\qquad
e_{23}-\bar\beta_{2,3}e_{32}.
$$
Hence $\bx\otimes y\in \im \rhog$ if and only if
\beq\label{1.14}\bx_1y_2=\al_{1,2} \bx_2 y_1, \quad \bx_1y_3=\al_{1,3} \bx_3 y_1, \quad \bx_2y_3=\beta_{2,3} \bx_3y_2.\eeq
By Lemma \ref{1.6}, there are four possibilities:
\begin{enumerate}
\item $x=0$ or $y=0$;
\item  $(\bx_1,y_1)=(c_1t,c_1\al_1)$, $(\bx_2,y_2)=(0,0)$ and $(\bx_3,y_3)=(c_3 t, c_3\al_3)$ for $t\ne 0$ and $c_1,c_3\in \CC^*$;
\item $(\bx_1,y_1)=(c_1t,c_1\al_1)$, $(\bx_2,y_2)=(c_2t,c_2\al_2)$ and $(\bx_3,y_3)=(0,0)$ for $t\ne 0$ and $c_1,c_2\in \CC$;
\item $(\bx_1,y_1)=(0,0)$, $(\bx_2,y_2)=(c_2t,c_2\beta_2)$ and $(\bx_3,y_3)=(c_3t,c_3\beta_3)$ for $t\ne 0$ and $c_2,c_3\in \CC^*$.
\end{enumerate}

The partial transpose $\varrho$ of $\rhog$ is given by
$$
\varrho=\left(\begin{array}{ccc|ccc|ccc}
r& \cdot&\cdot&  \cdot& \al_{1,2}^{-1}&\cdot &   \cdot& \cdot&2\al_{1,3}^{-1}\\
\cdot& 1&\cdot&  \cdot& \cdot&\cdot   &\cdot& \cdot&\cdot\\
\cdot& \cdot&2 &\cdot& \cdot&\cdot & \cdot& \cdot&\cdot\\
\hline
\cdot&\cdot&\cdot&  1& \cdot&\cdot &\cdot& \cdot&\cdot\\
\al_{1,2}& \cdot&\cdot &   \cdot& r&\cdot &\cdot& \cdot&\beta_{2,3}^{-1}\\
\cdot& \cdot&\cdot &   \cdot& \cdot&1 &\cdot& \cdot&\cdot\\
\hline
\cdot& \cdot&\cdot &\cdot& \cdot&\cdot &2& \cdot&\cdot\\
\cdot& \cdot&\cdot &\cdot& \cdot&\cdot &\cdot& 1&\cdot\\
2\al_{1,3}& \cdot&\cdot &\cdot& \beta_{2,3}&\cdot &\cdot& \cdot&r
\end{array}\right).
$$
We see that $\varrho$ is a positive \semi matrix of corank one if and only if
its $(e_{11}, e_{22}, e_{33})$-principal submatrix
$$D_3^{\al,\beta}(r)=\left(\begin{matrix}
r&\al_{1,2}^{-1}&2\al_{1,3}^{-1}\\
\al_{1,2} & r&\beta_{2,3}^{-1}\\
2\al_{1,3} & \beta_{2,3} &r
\end{matrix}\right).$$
is such a matrix.

Now we let $(\al_1, \al_2, \al_3)=(1, \al, \al^2)$ for $\alpha\ne \pm 1\in \CC$ with $|\al|=1$ and let $(\beta_1,\beta_2,\beta_3)=(1,1,1)$. Then
$$D_3^{\al,\beta}(r)=\left(\begin{matrix}
r&\al^{-1}&2\al^{-2}\\
\al & r&1\\
2\al^2 & 1 &r
\end{matrix}\right).$$
Let $\hat{r}$ be the largest zero of the equation \[\det D_3(r)=r^3-6r+2\al+2\al^{-1}=0.\] Then {$\hat{r}> 0$} and
$\hat{r}$ is a simple root since $|\al+\al^{-1}|<2\sqrt{2}$. So, $D_3(\hat{r})$ is a positive \semi matrix of corank one.
By direct computation, the kernel vector is given by
$$
(-2\bar \al^2 +\hat r^{-1}\bar\al,\ 2\bar\al\hat r^{-1} -1,\ \hat r-\hat r^{-1})^\ttt\in\mathbb C^3.
$$
Hence $x\otimes y\in \im \varrho$ if and only if
$$(-2\al^2+\hat{r}^{-1}\al)x_1y_1+(2\al \hat{r}^{-1}-1)x_2y_2+(\hat{r}-\hat{r}^{-1})x_3y_3=0.$$

It is easy to see that when $\alpha\ne \pm 1$, $x\otimes y\notin \im \varrho$ if (ii), (iii) or (iv) holds. Indeed, if (ii) holds, then
$$|c_1|^2(-2\al^2+\hat{r}^{-1}\al) +|c_3|^2(\hat{r}-\hat{r}^{-1})\al^2=0.$$
But when $\alpha\ne \pm 1$, the two complex numbers $(-2\al^2+\hat{r}^{-1}\al)$ and $(\hat{r}-\hat{r}^{-1})\al^2$ are
linearly independent over $\mathbb{R}$. Therefore $c_1=c_3=0$. The arguments for the cases (iii) and (iv) are similar.
Hence if $x\otimes y\in \im \varrho$ and $\bx\otimes y\in \im \rhog$, then $x=0$ or $y=0$.
Therefore $\varrho$ above is a PPT entangled edge state of corank one.

\section{$4\otimes 4$ edge states of corank one}\label{4x4}

Let $\al_i, \beta_i\in \CC$ with $|\alpha_i|=|\beta_i|=1$ for $1\le i\le 4$.
As above, we put $\al_{i,j}=\al_i^{-1}\al_j$ and $\beta_{i,j}=\beta_i^{-1}\beta_j$
and assume $\al_{i,j}\ne \beta_{i,j}$ for $1\le i<j\le 4$.

Let $\rhog$ be the $16\times 16$ matrix defined as follows:
\begin{enumerate}
\item The $(e_{12}, e_{21})$-principal submatrix is
$P_2(\al_{1,2})$.
\item The $(e_{13}, e_{31})$-principal submatrix is
$2P_2(\al_{1,3})$.
\item The $(e_{14}, e_{23}, e_{32}, e_{41})$-principal submatrix is
$P_4(1,\al_{2,3},\al_{1,4})$.
\item The $(e_{24}, e_{42})$-principal submatrix is
$2P_2(\beta_{2,4})$.
\item The $(e_{34}, e_{43})$-principal submatrix is
$P_2(\beta_{3,4})$.
\item The $(e_{11},e_{22},e_{33}, e_{44})$-principal submatrix is $rI_4$ for $r>1$ to be determined later.
\item All the other entries are zero.
\end{enumerate}

Namely, $\rhog$ is given by
$$
\begin{scriptsize}
\left(\begin{array}{cccc|cccc|cccc|cccc}
r& \cdot& \cdot& \cdot& \cdot& \cdot& \cdot& \cdot& \cdot& \cdot& \cdot& \cdot& \cdot& \cdot& \cdot& \cdot\\
\cdot& 1& \cdot& \cdot& \al_{1,2}& \cdot& \cdot& \cdot& \cdot& \cdot& \cdot& \cdot& \cdot& \cdot& \cdot& \cdot\\
\cdot& \cdot& 2& \cdot& \cdot& \cdot& \cdot& \cdot& 2\al_{1,3}& \cdot& \cdot& \cdot& \cdot& \cdot& \cdot& \cdot\\
\cdot& \cdot& \cdot& 2& \cdot& \cdot& 1& \cdot& \cdot& \cdot& \cdot& \cdot& \al_{1,4}& \cdot& \cdot& \cdot\\
\hline
\cdot& \al_{1,2}^{-1}& \cdot& \cdot& 1& \cdot& \cdot& \cdot& \cdot& \cdot& \cdot& \cdot& \cdot& \cdot& \cdot& \cdot\\
\cdot& \cdot& \cdot& \cdot& \cdot& r& \cdot& \cdot& \cdot& \cdot& \cdot& \cdot& \cdot& \cdot& \cdot& \cdot\\
\cdot& \cdot& \cdot& 1& \cdot& \cdot& 2& \cdot& \cdot& \al_{2,3}& \cdot& \cdot& \cdot& \cdot& \cdot& \cdot\\
\cdot& \cdot& \cdot& \cdot& \cdot& \cdot& \cdot& 2& \cdot& \cdot& \cdot& \cdot& \cdot& 2\beta_{2,4}& \cdot& \cdot\\
\hline
\cdot& \cdot& 2\al_{1,3}^{-1}& \cdot& \cdot& \cdot& \cdot& \cdot& 2& \cdot& \cdot& \cdot& \cdot& \cdot& \cdot& \cdot\\
\cdot& \cdot& \cdot& \cdot& \cdot& \cdot& \al_{2,3}^{-1}& \cdot& \cdot& 2& \cdot& \cdot& \al_{2,3}^{-1}\al_{1,4}& \cdot& \cdot& \cdot\\
\cdot& \cdot& \cdot& \cdot& \cdot& \cdot& \cdot& \cdot& \cdot& \cdot& r& \cdot& \cdot& \cdot& \cdot& \cdot\\
\cdot& \cdot& \cdot& \cdot& \cdot& \cdot& \cdot& \cdot& \cdot& \cdot& \cdot& 1& \cdot& \cdot& \beta_{3,4}& \cdot\\
\hline
\cdot& \cdot& \cdot& \al_{1,4}^{-1}& \cdot& \cdot& \cdot& \cdot& \cdot& \al_{2,3}\al_{1,4}^{-1}& \cdot& \cdot& 2& \cdot& \cdot& \cdot\\
\cdot& \cdot& \cdot& \cdot& \cdot& \cdot& \cdot& 2\beta_{2, 4}^{-1}& \cdot& \cdot& \cdot& \cdot& \cdot& 2& \cdot& \cdot\\
\cdot& \cdot& \cdot& \cdot& \cdot& \cdot& \cdot& \cdot& \cdot& \cdot& \cdot& \beta_{3, 4}^{-1}& \cdot& \cdot& 1& \cdot\\
\cdot& \cdot& \cdot& \cdot& \cdot& \cdot& \cdot& \cdot& \cdot& \cdot& \cdot& \cdot& \cdot& \cdot& \cdot& r
\end{array}\right)
\end{scriptsize}
$$
which is a $4\times 4$ matrix whose entries are also $4\times 4$ matrices.

By construction, it is easy to see that $\rhog$ is positive \semi of corank 5.
The partial transpose $\varrho$ is
$$
\begin{scriptsize}
\varrho=\left(\begin{array}{cccc|cccc|cccc|cccc}
r&\cdot&\cdot&\cdot&\cdot&\al_{1,2}^{-1}&\cdot&\cdot&\cdot&\cdot&2\al_{1,3}^{-1}&\cdot&\cdot&\cdot&\cdot&\al_{1,4}^{-1}\\
\cdot&1&\cdot&\cdot&\cdot&\cdot&\cdot&\cdot&\cdot&\cdot&\cdot&\cdot&\cdot&\cdot&\cdot&\cdot\\
\cdot&\cdot&2&\cdot&\cdot&\cdot&\cdot&1&\cdot&\cdot&\cdot&\cdot&\cdot&\cdot&\cdot&\cdot\\
\cdot&\cdot&\cdot&2&\cdot&\cdot&\cdot&\cdot&\cdot&\cdot&\cdot&\cdot&\cdot&\cdot&\cdot&\cdot\\
\hline
\cdot&\cdot&\cdot&\cdot&1&\cdot&\cdot&\cdot&\cdot&\cdot&\cdot&\cdot&\cdot&\cdot&\cdot&\cdot\\
\al_{1,2}&\cdot&\cdot&\cdot&\cdot&r&\cdot&\cdot&\cdot&\cdot&\al_{2,3}^{-1}&\cdot&\cdot&\cdot&\cdot&2\beta_{2,4}^{-1}\\
\cdot&\cdot&\cdot&\cdot&\cdot&\cdot&2&\cdot&\cdot&\cdot&\cdot&\cdot&\cdot&\cdot&\cdot&\cdot\\
\cdot&\cdot&1&\cdot&\cdot&\cdot&\cdot&2&\cdot&\cdot&\cdot&\cdot&\cdot&\cdot&\cdot&\cdot\\
\hline
\cdot&\cdot&\cdot&\cdot&\cdot&\cdot&\cdot&\cdot&2&\cdot&\cdot&\cdot&\cdot&\al_{2,3}\al_{1,4}^{-1}&\cdot&\cdot\\
\cdot&\cdot&\cdot&\cdot&\cdot&\cdot&\cdot&\cdot&\cdot&2&\cdot&\cdot&\cdot&\cdot&\cdot&\cdot\\
2 \al_{1,3}&\cdot&\cdot&\cdot&\cdot&\al_{2,3}&\cdot&\cdot&\cdot&\cdot&r&\cdot&\cdot&\cdot&\cdot&\beta_{3,4}^{-1}\\
\cdot&\cdot&\cdot&\cdot&\cdot&\cdot&\cdot&\cdot&\cdot&\cdot&\cdot&1&\cdot&\cdot&\cdot&\cdot\\
\hline
\cdot&\cdot&\cdot&\cdot&\cdot&\cdot&\cdot&\cdot&\cdot&\cdot&\cdot&\cdot&2&\cdot&\cdot&\cdot\\
\cdot&\cdot&\cdot&\cdot&\cdot&\cdot&\cdot&\cdot&\al_{2,3}^{-1}\al_{1,4}&\cdot&\cdot&\cdot&\cdot&2&\cdot&\cdot\\
\cdot&\cdot&\cdot&\cdot&\cdot&\cdot&\cdot&\cdot&\cdot&\cdot&\cdot&\cdot&\cdot&\cdot&1&\cdot\\
\al_{1,4}&\cdot&\cdot&\cdot&\cdot&2 \beta_{2,4}&\cdot&\cdot&\cdot&\cdot&\beta_{3,4}&\cdot&\cdot&\cdot&\cdot&r
\end{array}\right).\end{scriptsize}
$$

The matrix $\varrho$ is the direct sum of the following:
\begin{enumerate}
\item the $(e_{12}, e_{21},e_{34},e_{43})$-principal submatrix is the $4\times 4$ identity matrix $I_4$;
\item the $(e_{14}, e_{23}, e_{32},e_{41})$-principal submatrix is  $2I_{4}$;
\item the $(e_{13},e_{24})$-principal submatrix is the positive definite matrix $Q_{2}(1)$;
\item the $(e_{31},e_{42})$-principal submatrix is the positive definite matrix $Q_{2}(\al_{2,3}\al_{1,4}^{-1})$;
\item the $(e_{11},e_{22}, e_{33},e_{44})$-principal submatrix is the hermitian matrix
\[D_4^{\al,\beta}(r)= \left(\begin{array}{cccc}
r&\al_{1,2}^{-1}&2\al_{1,3}^{-1}&\al_{1,4}^{-1}\\
\al_{1,2}&r&\al_{2,3}^{-1}&2\beta_{2,4}^{-1}\\
2 \al_{1,3}&\al_{2,3}&r&\beta_{3,4}^{-1}\\
\al_{1,4}&2 \beta_{2,4}&\beta_{3,4}&r
\end{array}\right). \]
\end{enumerate}
Therefore, the matrix $\varrho$ is positive \semi of corank one if and only if $D_4^{\al,\beta}(r)$ is positive \semi of corank one.

Now, we choose the parameters $\al$ and $\beta$ by
\[ (\alpha_1,\alpha_2,\alpha_3,\alpha_4) = (1, e^{\frac{\pi i }{3}}, e^{\frac{2\pi i }{3}}, -1),
\qquad
(\beta_1,\beta_2,\beta_3,\beta_4)=(1,1,1,1).\]
Then we have
\[D_4^{\al,\beta}(r)= \left(\begin{array}{cccc}
r                    &e^{-\frac{\pi i}{3}} &2e^{-\frac{2\pi i}{3}}&-1\\
e^{\frac{\pi i }{3}} &r                    & e^{-\frac{\pi i}{3}} & 2\\
2e^{\frac{2\pi i}{3}}&e^{\frac{\pi i }{3}} & r                    &1\\
-1                   &2                    & 1                    &r
\end{array}\right). \]
We also choose $r$ to be the largest root $\hat{r}$ of the equation
\beq\label{eq:char4}\det(D_4^{\al,\beta}(r))=  r^4- 12 r^2 +6 r+17=0.\eeq
It is easy to see that $\hat{r}$ is a simple zero of \eqref{eq:char4}.
Therefore $D_4^{\al,\beta}(\hat{r})$ is positive \semi of corank 1.
By direct computation, we see that the kernel of $D_4^{\al,\beta}(\hat{r})$ is generated by
the vector
\[
(2 \hat{r}^2 - 4 i \sqrt{3} \hat{r} -5 + 5 i \sqrt{3} ,\ - 4 \hat{r}^2-  2 i \sqrt{3} \hat{r} + 16 + 4i \sqrt{3} ,\
- 2 \hat{r}^2 + 4 \hat{r}-3 - 3 i \sqrt{3} ,\  2 \hat{r}^3 - 12 \hat{r} + 8 ). \]
Hence $x\otimes y \in \im \varrho$ if and only if the vector $(x_1y_1, x_2y_2, x_3y_3, x_4y_4)$ is orthogonal
to this vector.

By definition of $\rhog$, one can see that $\bx\otimes y\in \im \rhog$
if and only if $v_i=(\bx_i,y_i)$ satisfy Lemma \ref{1.6}. Now, we can numerically check that
if $\bx\otimes y\in \im \rhog$ and $x\otimes y \in \im \varrho$, then $x=0$ or $y=0$.
Therefore, we conclude that $\varrho$ is a PPT entangled edge state of corank one.

\section{$n\otimes n$ edge states of corank one for $n\ge 3$}\label{nxn}

We generalize the above construction for any $n\ge 3$.
We fix $\al_1=1$, $\beta_n=1$ and for $2\le i\le n$, let
$\al_{i}, \beta_i \in \CC$ with $|\al_i|=1$ and $|\beta_i|=1$. Let
$$\al_{i,j}=\al_i^{-1}\al_j,\quad  \beta_{i,j}=\beta_i^{-1}\beta_j$$
as before, and we assume $\al_{i,j}\ne \beta_{i,j}$ for $1\le i<j\le n$ as in Lemma \ref{1.6}.

Let $\rhog$ be the $n^2\times n^2$ matrix defined as follows:
\begin{enumerate}[leftmargin=30pt]
\item The $(e_{12}, e_{21})$-principal submatrix is
$P_2(\al_{1,2})$.
\item The $(e_{13}, e_{31})$-principal submatrix is
$2P_2(\al_{1,3})$.
\item The $(e_{1,k}, e_{2,k-1}, e_{3,k-2},\cdots,e_{\lfloor \frac{k}2\rfloor,k+1-\lfloor\frac{k}2\rfloor},
e_{k+1-\lfloor\frac{k}2\rfloor,\lfloor\frac{k}2\rfloor},e_{k+2-\lfloor \frac{k}2\rfloor,\lfloor \frac{k}2\rfloor-1},\cdots, e_{k,1})$-principal
submatrix is
$$P_{2\lfloor \frac{k}2\rfloor}(1,1,\cdots,1,
\al_{\lfloor \frac{k}2\rfloor,k+1-\lfloor\frac{k}2\rfloor},\al_{\lfloor \frac{k}2\rfloor-1,k+2-\lfloor\frac{k}2\rfloor},\cdots, \al_{1,k})$$
for $4\le k\le n$.
\item The $(e_{l,n}, e_{l+1,n-1},\cdots, e_{l+\lfloor \frac{n-l+1}2\rfloor -1,n+1-\lfloor \frac{n-l+1}2\rfloor},
e_{n+1-\lfloor \frac{n-l+1}2\rfloor,l+\lfloor \frac{n-l+1}2\rfloor -1},\cdots, e_{n,l})$-principal submatrix is
$$P_{2\lfloor \frac{n-l+1}2\rfloor}(1,1,\cdots,1,\beta_{l+\lfloor \frac{n-l+1}2\rfloor -1,n+1-\lfloor \frac{n-l+1}2\rfloor},
\beta_{l+\lfloor \frac{n-l+1}2\rfloor -2,n+2-\lfloor \frac{n-l+1}2\rfloor},\cdots,\beta_{l,n})$$ for $2\le l\le n-3$.
\item The $(e_{n-2,n}, e_{n,n-2})$-principal submatrix is
$2P_2(\beta_{n-2,n})$ for $n>3$.
\item The $(e_{n-1,n}, e_{n,n-1})$-principal submatrix is
$P_2(\beta_{n-1,n})$.
\item The $(e_{11},e_{22},e_{33}, \cdots,e_{nn})$-principal submatrix is $rI_n$ for $r>1$ to be determined later.
\item All the other entries are zero.
\end{enumerate}

The matrix $\rhog$ is positive \semi of corank $2n-3$ such that $\bx\otimes y\in \im \rhog$
if and only if $v_i=(\bx_i,y_i)$ satisfy Lemma \ref{1.6}.
It is easy to check that the partial transpose $\varrho$ of $\rhog$ is the direct sum of the following:
\begin{enumerate}[leftmargin=30pt]
\item The $(e_{12}, e_{21},e_{n-1,n},e_{n,n-1})$-principal submatrix is the $4\times 4$ identity matrix $I_4$;
\item The $(e_{1,n}, e_{23}, e_{32},e_{34},e_{43},\cdots, e_{n-2,n-1},e_{n-1,n-2},e_{n,1})$-principal submatrix is  $2I_{2n-4}$;
\item For $3\le k<n$, the $(e_{1,k},e_{2,k+1},e_{3,k+2},\cdots,e_{n-k+1,n})$-principal submatrix
is the positive definite matrix $$Q_{n-k+1}(1,1,\cdots,1);$$
\item For $3\le k<n$, the $(e_{k,1},e_{k+1,2},e_{k+2,3},\cdots,e_{n,n-k+1})$-principal submatrix
is the positive definite matrix
\begin{align*}
Q_{n-k+1}(&\al_{2,k}\al_{1,1+k}^{-1}, \al_{3,k+1}\al_{2,2+k}^{-1},\cdots,
\al_{\lfloor \frac{n-k+3}2\rfloor,\lfloor \frac{n+k-1}2\rfloor}\al_{\lfloor \frac{n-k+1}2\rfloor,\lfloor \frac{n+k+1}2\rfloor}^{-1}, \\
&\beta_{\lfloor \frac{n-k+3}2\rfloor+1,\lfloor \frac{n+k-1}2\rfloor+1}
\beta_{\lfloor \frac{n-k+1}2\rfloor+1,\lfloor \frac{n+k+1}2\rfloor+1}^{-1}, \cdots,
\beta_{n-k,n-2}\beta_{n-k-1,n-1}^{-1},\beta_{n-k+1,n-1}\beta_{n-k,n}^{-1});
\end{align*}
\item The $(e_{11},e_{22}, \cdots,e_{nn})$-principal submatrix is
the $n\times n$ hermitian matrix $D_n^{\al,\beta}(r)$ defined as follows:
\begin{enumerate}
\item The first column is $(r, \al_{2},2\al_{3},\al_{4}, \al_{5},\cdots,\al_{n})$;
\item The first row is $(r,\bal_{2},2\bal_{3},\bal_{4}, \bal_{5},\cdots,\bal_{n})$;
\item The last column is $(\bal_{n},\beta_2,\beta_3,\cdots,\beta_{n-3},2\beta_{n-2},\beta_{n-1},r)$;
\item The last row is $(\al_n,\bar\beta_2,\bar\beta_3,\cdots,\bar\beta_{n-3}, 2\bar\beta_{n-2},\bar\beta_{n-1},r)$;
\item The diagonal entries (i.e. $(i,j)$th entries with $|i-j|=0$) are all $r$;
\item The $(i,j)$th entries with $|i-j|=1$ and $i+j\le n+1$ are $\al_{i,j}^{-1}=\al_i\bal_j$;
\item The $(i,j)$th entries with $|i-j|=1$ and $i+j> n+1$ are $\beta_{i,j}^{-1}=\beta_i\bar\beta_j$;
\item The $(i,j)$th entries with $|i-j|=2$ and $i+j\le n+1$ are $\al_{i,j}^{-1}$ except the $(1,3)$, $(3,1)$th entries;
\item The $(i,j)$th entries with $|i-j|=2$ and $i+j> n+1$ are $\beta_{i,j}^{-1}$ except the $(n-2,n)$, $(n,n-2)$th entries;
\item All other entries are 0.
\end{enumerate}
\end{enumerate}

Note that the matrix $D_n^{\al,\beta}(r)$ looks like
$$
\left(
\begin{matrix}
r        &\bar\al_2      &2\bar\al_3     &\bar\al_4       &\bar\al_5       &\cdots    &\cdots     &\bar\al_{n-1}     &\bar\al_n\\
\al_2    &r              &\bar\al_3\al_2 &\bar\al_4\al_2  &0               &\cdots    &\cdots     &0                 &\beta_2\\
2\al_3   &\bar\al_2\al_3 &r              &\bar\al_4\al_3  &\bar\al_5\al_2  &\ddots    &           &0                 &\beta_3\\
\al_4    &\bar\al_2\al_4 &\bar\al_3\al_4 &r               &\bar\al_5\al_4  &          &\ddots     &0                 &\beta_4\\
\al_5    &0              &\bar\al_3\al_5 &\bar\al_4\al_5  &r               &          &           &0                 &\beta_5\\
\vdots   &\vdots         &\ddots         &                &                &\ddots    &           &                  &\vdots\\
\vdots   &\vdots         &               &\ddots          &                &          &r          &\bar\beta_{n-1}\beta_{n-2}&2\beta_{n-2}\\
\al_{n-1}&0              &0              &0               &0               &          &\bar\beta_{n-2}\beta_{n-1}  &r      &\beta_{n-1}\\
\al_n    &\bar\beta_2    &\bar\beta_3    &\bar\beta_4     &\bar\beta_5     &\cdots    &2\bar\beta_{n-2}&\bar\beta_{n-1}&r
\end{matrix}
\right).
$$


Hence, $\varrho$ is a positive \semi matrix of corank one if and only if $D_n^{\al,\beta}(r)$ is such a matrix.
As before, we take $r$ to be the largest root $\hat{r}$ of the polynomial $\mathrm{det}D_n^{\al,\beta}(r)$.
If furthermore $\hat{r}$ is a simple root of $\mathrm{det}D_n^{\al,\beta}(r)$, then $D_n^{\al,\beta}(\hat{r})$
is a positive \semi matrix of corank one. It now amounts to checking the following to construct an edge state of corank one.

\begin{proposition}\label{prop:orth}
Let $\al_i$ and $\beta_i$ be as above and assume that the largest root $\hat{r}$ of $\mathrm{det}D_n^{\al,\beta}(r)=0$
is a simple root. Let $w=(w_1, \cdots, w_n)\in \mathbb{C}^n$ be a nonzero vector in the kernel of $D_n^{\al,\beta}(\hat{r})$.
Then the matrix $\varrho$ constructed above is an edge state of corank one if
\begin{itemize}[leftmargin=20pt]
\item[{\rm ($\star$)}]
the following vectors are \emph{not} orthogonal to $w$:
\begin{enumerate}
\item $(u_1, u_2\al_2, \cdots, u_q\al_{q},0,\cdots,0)$ where $q\le \frac{n+1}{2}$ and $u_i\ge 0$
for all $i$ but not all zero;
\item $(0,\cdots,0,u_p\al_{p},u_{p+1}\al_{p+1},\cdots,u_{n-q+1}\al_{n-q+1},0,0,\cdots,0,u_{q}\al_{q},0,\cdots,0)$
where $u_i\ge 0$ for all $i$ and $u_p\ne 0$, $u_{q}\ne 0$ with $p\le n-q+1 < \frac{n+1}2$;
\item $(0,\cdots, 0, u_{p}\beta_{p},0,\cdots,0, u_{n-p+2}\beta_{n-p+2},u_{n-p+3}\beta_{n-p+3},\cdots,u_{q}\beta_{q},0,\cdots, 0)$
where $u_i\ge 0$ for all $i$ and $u_{p}\ne 0$, $u_{q}\ne 0$ with $\frac{n+3}2 < n-p+2\le q$;
\item $(0,\cdots,0,u_{p}\beta_{p},\cdots,u_n\beta_n)$ where $u_i\ge 0$ for all $i$ but not all zero and $p \ge \frac{n+1}{2}$.
\end{enumerate} Here $u_i$ is at the $i$th place.
\end{itemize}
\end{proposition}
\begin{proof}
It is straightforward from the construction of $\rhog$ and $\varrho$ above and Lemma \ref{1.6}. Here,
$u_j=|c_j|^2$ where $c_j$ are from Lemma \ref{1.6}.
\end{proof}

To check that the vector $w$ satisfies ($\star$), we may use the following lemma whose proof is straightforward and omitted.
\begin{lemma}\label{check}
For $z_1, \cdots, z_m\in \mathbb{C}$, there do not exist nonnegative (not all zero) real numbers $u_j$ for $j=1, \cdots, m$
such that $\sum_{j=1}^m u_jz_j=0$ if and only if $z_1, \cdots, z_m$ belong to a half-plane on the complex plane, i.e.
there exists a nonzero $h\in \mathbb{C}$ such that $\re(\bar{z}_jh)>0$ for all $j$, or equivalently, $z_j\ne 0$ for all $j$ and either
\[ \max_{j} \mathrm{Arg}(z_j) - \min_j \mathrm{Arg}(z_j) <\pi \quad\text{or}\quad \max_{j} \mathrm{Arg}(-z_j) - \min_j \mathrm{Arg}(-z_j) <\pi\]
holds, where $\mathrm{Arg}(z)$ is the principal value of the argument of $z$ with $-\pi< \mathrm{Arg}(z)\le\pi$.
\end{lemma}

For example, to check the vectors in (i) of Proposition \ref{prop:orth} are not orthogonal to $w$,
we check that the complex numbers $w_1, w_2\bar{\al}_2, \cdots, w_q\bar{\al}_q$ belong to a half-plane.

Note that as $\hat{r}$ and $w$ are determined algebraically by parameters $\al_i$ and $\beta_i$, the condition ($\star$)
and $\hat{r}$ being a simple root are open conditions, i.e. if the conditions are satisfied for some $\al_i$ and $\beta_i$,
then they also hold for all $\al_i'$ and $\beta_i'$ sufficiently close to $\al_i$ and $\beta_i$ respectively.
So, the set of tuples $(\al_2, \cdots, \al_n, \beta_2, \cdots \beta_{n-1})$ which yield PPT entangled edge states
of corank one is an open subset of ${U(1)}^{2n-3}$ where $U(1)$ denotes the circle group $\{z\in \CC\,|\, |z|=1\}$.
Hence if nonempty, our construction produces a $(2n-3)$-dimensional  family of PPT entangled edge states of corank one.

A program for checking Proposition \ref{prop:orth} has been implemented in Mathematica (available upon request).
The algorithm proceeds as follows:
\medskip
\begin{enumerate}[label={\bf Step \arabic*}:,leftmargin=60pt]
\item Given $\al_i$ and $\beta_i$, find the largest root $\hat{r}$ of the equation $\mathrm{det}D_n^{\al,\beta}(r)=0$.
\item Check that the dimension of the kernel of $D_n^{\al,\beta}(\hat{r})$ is one.
\item Find a nonzero vector $w=(w_1, \cdots, w_n)\in \mathbb{C}^n$ in the kernel of $D_n^{\al,\beta}(\hat{r})$
and check that $w_i\ne 0$ for all $i=1, \cdots, n$.
\item Check that the vector $w$ satisfies ($\star$) using Lemma \ref{check}.
\end{enumerate}
\medskip

We have checked that the conditions in Proposition \ref{prop:orth} are satisfied for all $3\le n\le 1000$ if we let
\beq \label{param}
(\al_1, \cdots, \al_n)= (1, e^{\frac{\pi i}{4}},e^{\frac{\pi i}{4}},\cdots, e^{\frac{\pi i}{4}}, -1)\quad \text{and}\quad
(\beta_1, \cdots, \beta_n)=(1,1,\cdots, 1).
\eeq
As remarked in the previous paragraph, we may now perturb $\al_i$ and $\beta_i$ so that
$\al_{i,j}\ne\beta_{i,j}$ for $1\le i<j\le n$ and the conditions in Proposition \ref{prop:orth} are still satisfied.
For example, we have also checked that the perturbation of the form
$
(\al_1, \cdots, \al_n)= (1, e^{\pi i \left(\frac{1}{4}+\frac{1}{10000}\right)},e^{\pi i \left(\frac{1}{4}+\frac{2}{10000}\right)},
\cdots, e^{\pi i \left(\frac{1}{4}+\frac{n-2}{10000}\right)}, -1)$
and $
(\beta_1, \cdots, \beta_n)=(1,1,\cdots, 1)
$
gives a valid result up to $n=346$. Although the perturbation of this form fails when $n>346$, one can find $\al_i$
and $\beta_i$ close to \eqref{param} which gives a PPT entangled edge state.

Therefore, we have the following

\begin{theorem}\label{1.2}
For $3\le n \le 1000$, there is a PPT entangled edge state in $M_n\otimes M_n$ of bi-rank $(n^2-1,n^2-2n+3)$.
\end{theorem}

We believe that the same construction with $\al_i$ and $\beta_i$ sufficiently close to \eqref{param} works in general.
We stopped at $n=1000$ because of the running time of the program. On an ordinary laptop computer with 2.9GHz processor,
it takes 3 minutes to check for $3\le n\le 200$, about an hour for $200< n\le 400$, about 3 hours for $400< n\le 600$, about 7 hours for $600< n\le 800$ and about 15 hours for $800< n\le 1000$. Based on this numerical evidence, we propose the following

\begin{conjecture}
For any $n\ge 3$, the open set of parameters $(\al_1, \cdots, \al_n, \beta_1, \cdots, \beta_n)$ which give
PPT entangled edge states of bi-rank $(n^2-1,n^2-2n+3)$ is nonempty.
\end{conjecture}

\section{Discussion}

In the previous section, we had to rely on computer computations because 
it is difficult to find an explicit formula for a kernel vector of $\varrho$. On the other hand, in \cite{kye-osaka}, the authors provide a construction of a $3\otimes 3$ edge state in which a kernel vector of $\varrho$ is fixed from the beginning. We could generalize this construction to the $4\otimes 4$ case, as we will explain briefly below.

We begin with
$$
A=
\left(\begin{matrix}
\alpha+\bar\alpha&-\alpha &-\bar\alpha & 0\\
-\bar\alpha&\alpha+\bar\alpha &0 &-\alpha\\
-\alpha &0&\alpha+\bar\alpha&-\bar\alpha\\
0&-\bar\alpha&-\alpha&\alpha+\bar\alpha
\end{matrix}\right),
$$
with $|\alpha|=1$. If $-\frac\pi 4<\mathrm{Arg} \alpha<\frac \pi 4$, then it is easy to check that $A$ is positive
of corank one with the kernel spanned by the vector $(1,1,1,1)^\ttt\in\mathbb C^4$.
We also consider
$$
B=
r\left(\begin{matrix}
\frac{1}{p}&0 &-\bar\alpha & 0\\
0&\frac{1}{p} &0 &-\bar\alpha\\
-\alpha &0&p&0\\
0&-\alpha&0&p
\end{matrix}\right)
+
r\left(\begin{matrix}
1&-1 &0 & 0\\
-1&1 &0 &0\\
0&0&1&-1\\
0&0&-1&1
\end{matrix}\right),
$$
where $p>0$ and $0<r<1$. Then $B$ is positive of corank one with the kernel spanned by the vector $(p,p,\alpha,\alpha)^\ttt\in\mathbb C^4$.

Now, we define the matrix $\varrho\in M_4(M_4)=M_4\otimes
M_4$ as follows:
\begin{enumerate}
\item
The $(e_{11}, e_{22}, e_{33}, e_{44})$-principal submatrix of $\varrho$ is given by $A$,
\item
The $(e_{14}, e_{23}, e_{32}, e_{41})$-principal submatrix of $\varrho^\Gamma$ is given by $B$,
\item
The diagonals of $\varrho$ are given by $\frac{1}{p}$ in the places $e_{12}, e_{24}, e_{31}, e_{43}$,
\item
The diagonals of $\varrho$ are given by $p$ in the places $e_{21}, e_{42}, e_{13}, e_{34}$.
\item
All the other entries are zero.
\end{enumerate}

Namely, $\rhog$ is given by
$$
\begin{scriptsize}
\left(\begin{array}{cccc|cccc|cccc|cccc}
  \al+\bal& \cdot& \cdot& \cdot& \cdot& \cdot& \cdot& \cdot& \cdot& \cdot& \cdot& \cdot& \cdot& \cdot& \cdot& \cdot\\
  \cdot& \frac{1}{p}& \cdot& \cdot& -\bal& \cdot& \cdot& \cdot& \cdot& \cdot& \cdot& \cdot& \cdot& \cdot& \cdot& \cdot\\
  \cdot& \cdot& p& \cdot& \cdot& \cdot& \cdot& \cdot& -\al& \cdot& \cdot& \cdot& \cdot& \cdot& \cdot& \cdot\\
  \cdot& \cdot& \cdot& \frac{r}{p}+r& \cdot& \cdot& -r& \cdot& \cdot& -r\bal& \cdot& \cdot& \cdot& \cdot& \cdot& \cdot\\
\hline
  \cdot& -\al& \cdot& \cdot& p& \cdot& \cdot& \cdot& \cdot& \cdot& \cdot& \cdot& \cdot& \cdot& \cdot& \cdot\\
  \cdot& \cdot& \cdot& \cdot& \cdot&  \al+\bal& \cdot& \cdot& \cdot& \cdot& \cdot& \cdot& \cdot& \cdot& \cdot& \cdot\\
  \cdot& \cdot& \cdot& -r& \cdot& \cdot& \frac{r}{p}+r& \cdot& \cdot& \cdot& \cdot& \cdot& -r\bal& \cdot& \cdot& \cdot\\
  \cdot& \cdot& \cdot& \cdot& \cdot& \cdot& \cdot& \frac{1}{p}& \cdot& \cdot& \cdot& \cdot& \cdot& -\bal& \cdot& \cdot\\
\hline
  \cdot& \cdot& -\bal& \cdot& \cdot& \cdot& \cdot& \cdot& \frac{1}{p}& \cdot& \cdot& \cdot& \cdot& \cdot& \cdot& \cdot\\
  \cdot& \cdot& \cdot& -r\al& \cdot& \cdot& \cdot& \cdot& \cdot& rp+r& \cdot& \cdot& -r& \cdot& \cdot& \cdot\\
  \cdot& \cdot& \cdot& \cdot& \cdot& \cdot& \cdot& \cdot& \cdot& \cdot&  \al+\bal& \cdot& \cdot& \cdot& \cdot& \cdot\\
  \cdot& \cdot& \cdot& \cdot& \cdot& \cdot& \cdot& \cdot& \cdot& \cdot& \cdot& p& \cdot& \cdot& -\al& \cdot\\
\hline
  \cdot& \cdot& \cdot& \cdot& \cdot& \cdot& -r\al& \cdot& \cdot& -r& \cdot& \cdot& rp+r& \cdot& \cdot& \cdot\\
  \cdot& \cdot& \cdot& \cdot& \cdot& \cdot& \cdot& -\al& \cdot& \cdot& \cdot& \cdot& \cdot& p& \cdot& \cdot\\
  \cdot& \cdot& \cdot& \cdot& \cdot& \cdot& \cdot& \cdot& \cdot& \cdot& \cdot& -\bal& \cdot& \cdot& \frac{1}{p}& \cdot\\
  \cdot& \cdot& \cdot& \cdot& \cdot& \cdot& \cdot& \cdot& \cdot& \cdot& \cdot& \cdot& \cdot& \cdot& \cdot&  \al+\bal
\end{array}\right),
\end{scriptsize}
$$
whose partial transpose $\varrho$ is given by
$$
\begin{scriptsize}
\left(\begin{array}{cccc|cccc|cccc|cccc}
  \al+\bal& \cdot& \cdot& \cdot& \cdot& -\al& \cdot& \cdot& \cdot& \cdot& -\bal& \cdot& \cdot& \cdot& \cdot& \cdot\\
  \cdot& \frac{1}{p}& \cdot& \cdot& \cdot& \cdot& \cdot& \cdot& \cdot& \cdot& \cdot& -r\al& \cdot& \cdot& \cdot& \cdot\\
  \cdot& \cdot& p& \cdot& \cdot& \cdot& \cdot& -r& \cdot& \cdot& \cdot& \cdot& \cdot& \cdot& \cdot& \cdot\\
  \cdot& \cdot& \cdot& \frac{r}{p}+r& \cdot& \cdot& \cdot& \cdot& \cdot& \cdot& \cdot& \cdot& \cdot& \cdot& \cdot& \cdot\\
\hline
  \cdot& \cdot& \cdot& \cdot& p& \cdot& \cdot& \cdot& \cdot& \cdot& \cdot& \cdot& \cdot& \cdot& -r\al& \cdot\\
  -\bal& \cdot& \cdot& \cdot& \cdot& \al+\bal& \cdot& \cdot& \cdot& \cdot& \cdot& \cdot& \cdot& \cdot& \cdot& -\al\\
  \cdot& \cdot& \cdot& \cdot& \cdot& \cdot& \frac{r}{p}+r& \cdot& \cdot& \cdot& \cdot& \cdot& \cdot& \cdot& \cdot& \cdot\\
  \cdot& \cdot& -r& \cdot& \cdot& \cdot& \cdot&\frac{1}{p}& \cdot& \cdot& \cdot& \cdot& \cdot& \cdot& \cdot& \cdot\\
\hline
  \cdot& \cdot& \cdot& \cdot& \cdot& \cdot& \cdot& \cdot& \frac{1}{p}& \cdot& \cdot& \cdot& \cdot& -r& \cdot& \cdot\\
  \cdot& \cdot& \cdot& \cdot& \cdot& \cdot& \cdot& \cdot& \cdot& rp+r& \cdot& \cdot& \cdot& \cdot& \cdot& \cdot\\
  -\al& \cdot& \cdot& \cdot& \cdot& \cdot& \cdot& \cdot& \cdot& \cdot& \al+\bal& \cdot& \cdot& \cdot& \cdot& -\bal\\
  \cdot& -r\bal& \cdot& \cdot& \cdot& \cdot& \cdot& \cdot& \cdot& \cdot& \cdot& p& \cdot& \cdot& \cdot& \cdot\\
\hline
  \cdot& \cdot& \cdot& \cdot& \cdot& \cdot& \cdot& \cdot& \cdot& \cdot& \cdot& \cdot& rp+r& \cdot& \cdot& \cdot\\
  \cdot& \cdot& \cdot& \cdot& \cdot& \cdot& \cdot& \cdot& -r& \cdot& \cdot& \cdot& \cdot& p& \cdot& \cdot\\
  \cdot& \cdot& \cdot& \cdot& -r\bal& \cdot& \cdot& \cdot& \cdot& \cdot& \cdot& \cdot& \cdot& \cdot& \frac{1}{p}& \cdot\\
  \cdot& \cdot& \cdot& \cdot& \cdot& -\bal& \cdot& \cdot& \cdot& \cdot& -\al& \cdot& \cdot& \cdot& \cdot& \al+\bal
\end{array}\right).\end{scriptsize}
$$

It is straightforward to check that both of $\varrho$ and $\rhog$ are positive when $p>0$ and $0<r<1$ and $\al$ is a complex number with $|\al|=1$ and $-\frac\pi 4<\mathrm{Arg} \alpha<\frac \pi 4$.
We also see that $\varrho$ has the corank one with the kernel spanned by the vector
$e_{11}+e_{22}+e_{33}+e_{44}$. On the other hand, the partial transpose $\varrho^\Gamma$
has the corank $5$ with the kernel spanned by the vectors
$p e_{12} +\alpha e_{21}$,
$p e_{24} +\alpha e_{42}$,
$p e_{31} +\alpha e_{13}$,
$p e_{43} +\alpha e_{34}$ and
$p e_{14} + p e_{23} +\alpha e_{32}+\alpha e_{41}$,
from which one may check easily that $\varrho$ is a PPT entangled edge state.
It would be very nice if this method works for any $n\ge 5$ to get PPT entangled edge states with exact formulae.

\section{Acknowledgment}
JC was supported by Korea NRF grant 2018R1C1B6005600. YHK was partially supported by Korea NRF grants 2017R1E1A1A03070694 and 2017R1A5A1015626.
SHK was partially supported by Korea NRF grant 2017R1A2B4006655.


\end{document}